\documentclass[11pt, letterpaper]{article}
\usepackage{amsthm}
\usepackage{amsmath}
\usepackage{amsfonts}
\usepackage{amssymb}
\usepackage{algorithm,algorithmic}
\usepackage{latexsym}
\usepackage{xspace}
\usepackage{graphicx}
\usepackage{hyperref}
\usepackage{mdwlist}
\usepackage{fullpage}
\usepackage{color}

\newtheorem{theorem}{Theorem}

\newtheorem{lemma}{Lemma}

\newtheorem{proposition}[theorem]{Proposition}

\newtheorem{corollary}[theorem]{Corollary}

\newcommand{\wst}{w_{s\setminus t}}

\newcommand{\ie}{{\it i.e.}}

\def\epsilon{\varepsilon}
\def\xstar{x^\star}
\def\wstar{w^\star}
\def\xcirc{x^\circ}

\def\showauthornotes{1}
\ifnum\showauthornotes=1
\newcommand{\Authornote}[2]{{\sf\small\color{red}{[#1: #2]}}}
\else
\newcommand{\Authornote}[2]{}
\fi

\title{Bargaining for Revenue Shares on Tree Trading Networks}

\author{Arpita Ghosh\thanks{Work done while the author was at Yahoo! Research.}\\
{Dept. of Information Science}\\
{Cornell University}\\
{Ithaca, NY}\\
arpitaghosh@cornell.edu
\and
Satyen Kale$^*$\\
IBM Watson\\
Yorktown Heights, NY\\
sckale@us.ibm.com
\and 
Kevin Lang\\
{Yahoo! Research}\\
{Sunnyvale, CA}\\
langk@yahoo-inc.com
\and
Benjamin Moseley\thanks{Work done while the author was at Yahoo! Research and the University of Illinois. Partially supported by NSF grant CCF-1016684. }\\
{TTI-Chicago}\\
{Chicago, IL}\\
moseley@ttic.edu
}

\begin{document}

\maketitle

\begin{abstract}
We study trade networks with a {\em tree} structure, where a seller with a single indivisible good is connected to buyers, each with some value for the good, via a unique path of intermediaries. Agents in the tree make multiplicative revenue share offers to their parent nodes, who choose the best offer and offer part of it to their parent, and so on; the winning path is determined by who finally makes the highest offer to the seller. In this paper, we investigate how these revenue shares might be set via a natural bargaining process between agents on the tree, specifically, {\em egalitarian} bargaining between endpoints of each edge in the tree.
We investigate the fixed point of this system of bargaining equations and prove various desirable for this solution concept, including (i) existence, (ii) uniqueness, (iii) efficiency, (iv) membership in the core, (v) strict monotonicity, (vi) polynomial-time computability to any given accuracy.
Finally, we present numerical evidence that asynchronous dynamics with randomly ordered updates always converges to the fixed point, indicating that the fixed point shares might arise from decentralized bargaining amongst agents on the trade network.

\end{abstract}


\section{Introduction}

Motivated by applications to ad exchanges such as the Yahoo!'s Right Media Exchange~\cite{YahooRM}, we consider a theoretical model of trade networks which take the form of a rooted tree.  In this model, publishers selling impressions can be connected via a string of {\em intermediary} ad-networks \cite{FeldmanMMP10} to advertisers interested in buying these impressions at the leaf nodes. These intermediaries want a cut of the surplus generated when a trade facilitated by them occurs. Typically, these cuts are specified as multiplicative {\em revenue shares or cuts} on edges that link a pair of entities. 
In practice, the value of these revenue shares would be set by business negotiations between the entities.  A natural theoretical question, is what constitutes a reasonable set of values for these revenue shares. Of course, a complete solution to this problem would require analyzing a very complex setting with advertisers and intermediaries optimizing over multiple heterogeneous impressions and publishers in a network setting: in this paper, we take the first steps towards understanding this problem by analyzing the sale of a single impression.

In our model, each buyer makes an offer to pay its parent intermediary in the tree a revenue share in the form of some fraction of its value for the item, \ie, for being matched to the seller;  the parent intermediary chooses the highest offer from all the buyers it is connected to. Each such intermediary then makes an offer to its parent, who selects the highest offer, and so on. Finally, the seller selects the highest offer it receives from its children in the tree, which determines the winning buyer. 
Given the tree structure and buyer values, the revenue shares completely specify the winning path and all winners' payoffs. 

The two-player bargaining problem widely studied in cooperative game theory, where the seller and a single buyer with value $v$ must fairly divide the value $v$ generated from their trade, is a special case of this setting: in the simplest version, the bargaining solution is to split the value $v$ equally amongst the two agents. Now consider a tree network where the child and parent nodes on each edge bargain about the revenue share using two-player bargaining. Here, the child might want to offer a revenue share greater than $1/2$ for two reasons:  first, the parent node might have other children to trade with that this node needs to beat out.
Second, and more unique to our setting, even if the child does beat out its siblings, the {\em parent} may not be able to make an adequately large offer to beat out {\em its} siblings higher up in the tree, and so on--- if this happens, neither the parent nor the child belongs to the winning path, and the value actually {\em realized} by the child is zero. So how much the child offers its parent, accounting for both of these effects, will depend on shares elsewhere in the tree, \ie, the revenue share negotiated on an edge depends on shares elsewhere in the tree. The key question we consider in this paper is whether there is a set of {\em mutually consistent} shares on all edges, and if yes, what kinds of outcomes it generates.

While there has been plenty of work on network bargaining problems building on the seminal work of Kleinberg and Tardos~\cite{KleinbergT08}, the model considered in those papers is unsuitable to our problem since the values that are being bargained on the edges are {\em exogenous}. In our setting, the value being split on an edge is {\em endogenous}, depending on the splits elsewhere in the tree. In the language of bargaining games, in the Kleinberg-Tardos model for network bargaining the feasible set for the bargaining problem on each edge is independent of shares on other edges (although the disagreement point is not), whereas in our setting the feasible set for an edge in the tree bargaining problem changes with shares elsewhere in the tree.

Such endogenous values on edges arise naturally in bargaining networks arising from trading settings, where there is competition for goods being sold. We consider the simplest possible version of this new bargaining model, which is bargaining on a tree. There are certainly many possible generalizations, but the goal of this paper is to analyze the simplest setting fully.
Thus, while many ad-hoc schemes can be proposed to compute these values, in this paper we investigate a natural bargaining game motivated by the fact that entities in the trade network negotiate revenue shares, and show that the outcome corresponding to the {\em unique fixed point} of this bargaining game on the trade network has many desirable properties.

\paragraph{Overview of Conceptual Contributions and Technical Results.} Our key conceptual contribution in this paper is the formulation of a bargaining game on the trading tree and a new solution concept for the game based on fixed points of the bargaining game. In this bargaining game, the nodes at the endpoints of each edge in the tree negotiate about how to split the value arriving at the child node according to two-player {\em egalitarian}, or proportional, bargaining~\cite{kalai}, given the splits elsewhere in the tree (Section \ref{sec:bargaining}). Each such two-player bargaining game gives us a (non-linear) equation for the revenue share on that edge, and a fixed point of the game is a simultaneous solution to this system of equations.


We first show a reduction from tree bargaining to path bargaining (Section \ref{sec:reduction}). We prove that for any tree and any set of buyer values, there is a reduced {\em path} such that every fixed point of the system of bargaining equations on the path corresponds to a fixed point of the system of bargaining equations on the tree, and vice versa.
This reduction to path bargaining allows us to analyze a smaller set of simultaneous equations, with one variable for each edge in the {\em path}, which we use to
prove the following set of results: 
\begin{enumerate}
\itemsep-2pt
\item {\em Existence and uniqueness:} We show that a fixed point to the bargaining equations always exists; further, the fixed point is unique.
\item {\em Efficiency:} The winner is always a buyer with highest value. (This is not true with Nash bargaining.)
\item {\em Core:} The payoffs given by the fixed point of the bargaining equations belong to the {\em core} of the natural cooperative game corresponding to our setting.
\item {\em Strict Monotonicity:} If the maximum value in the subtree rooted at a node in the winning path increases, the node's final payoff strictly increases as well.
\item {\em Computability:} The fixed point of the system of bargaining equations can be efficiently computed to any accuracy in polynomial time by a centralized algorithm.
\end{enumerate}\vspace{-2pt}
Finally, in Section \ref{sec:dynamics}, we present exhaustive numerical simulations indicating that {\em asynchronous dynamics}, where at each step a random edge in the tree renegotiates the share $x_e$ given the current shares in the remaining edges, converges rapidly to the fixed point-- this suggests that decentralized bargaining on edges should lead to the shares specified by this fixed point.

The outcome corresponding to the fixed point of the bargaining equations can be thought of as a solution concept for the corresponding cooperative game. A natural question is the suitability of other solution concepts such as the Shapley value or the nucleolus for our setting, or using Nash bargaining instead to define the solution concept: all these candidates seem to have some deficiency compared to our concept. We refer the interested reader to the appendixes~\ref{app:other} and \ref{app:nash} for a discussion.

\paragraph{Techniques.} Our results are based on several analytical and combinatorial techniques. First, we prove several structural properties that {\em any} fixed point solution, if one exists, must satisfy, which allows us to reduce the general tree bargaining problem to a structurally simpler path bargaining problem (see Section \ref{sec:reduction}). Next, for the path bargaining problem, we use analytic techniques to deduce certain monotonicity properties of any fixed point solution. These properties  directly give us uniqueness of the fixed point, assuming it exists. To show existence, we appeal to Brouwer's fixed point theorem by constructing a continuous mapping that is closely related to the bargaining equations. Our proofs of the core and strict monotonicity properties of the fixed point are again based on analytic techniques, and the use of an optimal substructure result for the fixed point which follows from our uniqueness result. Finally, by refining our monotonicity arguments quantitatively, we give an algorithm based on binary search to compute the fixed point to any accuracy, with running time that is polynomial in the number of nodes and the {\em logarithms} of the accuracy parameter and the gap between the highest and second highest values.

\paragraph{Related work.} The problem we study relates to many well-studied branches of the economics and computer science literature. The question of how agents on the winning path should split the generated value can be thought of as a {\em fair division} or revenue-sharing problem on which there is an extensive literature, albeit in settings different from ours; for an overview, see~\cite{moulin}. 
The work of Blume et al. \cite{BlumeEKT07} is perhaps the most similar in spirit to ours from this literature, though it looks at a different setting where traders set prices strategically and buyers and sellers react to these offers in a general trade network, and investigates subgame perfect Nash equilibria.

There is much recent work on bargaining in social networks, starting with the work of Kleinberg and Tardos~\cite{KleinbergT08}.
This work extends the classic two-player bargaining problem to a network where pairs of agents, instead of bargaining in isolation, can choose which neighbor to bargain with\footnote{For a nice survey of the literature on network exchange theory as well as two-player bargaining, see ~\cite{KleinbergT08,ChakrabortyK08} respectively.}.
A number of papers since \cite{KleinbergT08} have addressed computational and structural aspects of the network bargaining problem, as well as extensions to the model and dynamics; see~\cite{ChakrabortyK08,ChakrabortyKK09,AzarBCDP09,CelisDP10,AzarDJR10,Kanoria10}. While there are similarities between the network bargaining and our model, there are also fundamental differences: there, an outcome is a {\em matching} on the network, whereas we seek a {\em path}. More importantly, the values that are being bargained over on the edges are {\em exogenous} in their model, while in ours the value being split on an edge itself {\em depends} on the splits elsewhere in the tree: in the language of bargaining games, the feasible set for the bargaining problem on each edge is {\em independent} of shares on other edges in the network bargaining problem (although the disagreement point is not), whereas the feasible set for an edge in the tree bargaining problem {\em changes} with shares elsewhere in the tree.

\section{Model}
\label{s-model}
There is a seller selling a single item, buyers, each of whom derives some value from the item, and a number of intermediaries who assist in connecting buyers to the seller.
The trade network between these agents is given by a {\em rooted tree} $T$: the leaf nodes in $T$ (denoted generically by $l$) are the buyers, the root $r$ is the seller, and the internal nodes (denoted generically by $i$) are the intermediaries. We use $v_l$ to denote leaf $l$'s value for the item. 
The tree structure of the trade network means that each buyer has a {\em unique} path to the seller.
An {\em instance} $(T, \vec{v})$ of the tree bargaining problem is specified by the tree topology $T$, and the values $v_l$ at the leaves of $T$.

We use $e$ to denote edges 
 and $p$ to denote paths connecting the seller and a buyer in $T$. Given a path $p = \{r, i_1, \ldots, i_k, l\}$, we define the value of the path $v(p) = v_l$. For any two nodes $t_1$ and $t_2$ let $p_{t_1t_2}$ denote the unique path from $t_1$ to $t_2$ in the tree $T$. A child node in the tree makes an offer to its parent, who chooses the highest of these and offers part of it to {\em its} parent, and so on, as described next.

The endpoints of each edge $e=(t,s)$ in $T$ {\em split} the value that arrives at the child node $t$, specifying what portion of this value $t$ retains and what portion it is willing to pass up to $s$. We use $x_e$, where $x_e \in [0,1]$, to denote the multiplicative split or `revenue share' on edge $e$: if the potential value\footnote{We say potential value because this value is realized only if these nodes belong to the winning path.} arriving at $t$ is $w_t$, $t$ keeps $w_t(1-x_e)$ and passes up $w_t x_e$ to $s$. We use the multiplicative split $x_e$ rather than an additive split for convenience in correctly writing the bargaining equations. Note that the value of $x_e$ can, of course, depend on $w_t$, as well as the splits $x_{e'}$ on other edges $e' \in T$.

Given an instance $(T, \vec{v})$, an outcome consists of a {\em winning buyer} $l^*$, which also specifies the {\em winning path} $p^* = p_{l^*r}$, and a split of the value $v_{l^*}$ amongst the nodes on the winning path (including the leaf and the root).

The set of revenue shares $x_e$ completely specifies the outcome for an instance $(T,v)$ as follows. Every node in the tree, when presented with multiple children offering different payoffs, chooses to transact with the child that gives her the highest payoff.
Define the value reaching a non-leaf node $s \in T$, $w_s$, recursively as follows. Set $w_l = v_l$ for all leaves $l$, and let $C_{s}$ be the set of children of $s$ in $T$. Then, we have $w_s = \max_{t\in C_s} x_{ts}w_t$.

Let $t^*(s) = \arg\max_{t\in C_s} x_{ts}w_t$, with ties broken arbitrarily, denote the `winning child' of the parent node $s$.
The path  $p^* = (r, t^*(r), t^*(t^*(r)), \ldots, l^*)$ from root $r$ to leaf $l^*$ is the {\em winning path}, and $l^*$ is the {\em winning buyer}. 
The value $v_{l^*}$, generated by matching $l^*$ to $r$ is split among the nodes on $p^*$ using the revenue shares on edges of $p^*$. For all other nodes in the tree, the payoff is zero.


This setting can also be modeled as a cooperative game; we do this in Section \ref{sec:prop}.


\section{Bargaining on Trees}
\label{sec:bargaining}
Given an instance $(T,\vec{v})$, the splits $x_e$ on the edges $e \in T$ completely specify the outcome, namely who the winning agents are, and what payoffs they receive. How might these splits $x_e$ be determined?

We consider a bargaining-based determination of the shares $x_e$. We suppose that the agents corresponding to the endpoint of each edge negotiate according to two-player bargaining about how to split the value arriving at that edge.
The trading tree structure affects the two-player bargaining that takes place on each edge in two ways: first, the disagreement point for the parent node is determined by the offers it negotiates with its other children, and second, the feasible set of splits depends on the revenue share on the edge connecting the parent node to {\em its} parent, because the parent node must pass up this fraction of the value that it receives from the split.
Note that the revenue shares on these edges all influence each other, since the split of the value on one edge influences the bargaining power and therefore the split of the value on a different edge.

A natural choice for $x_e$, then, would be a {\em fixed point} to the system of bargaining equations, that is, a set of splits that are mutually consistent in the following sense: given the shares $x_{e'}$ on all remaining edges $e'$, the solution $y_e$ to the two-player bargaining problem on any edge $e$ with parameters specified by the remaining $x_{e'}$ is precisely $x_e$. It is not clear if such a fixed point exits, and even if it does, whether the final winner in a fixed point is the buyer with highest value.


\paragraph{Bargaining equations.}
The egalitarian, or proportional, bargaining solution~\cite{kalai} for the two-player bargaining problem on the edge $e= (t,s)$, given the shares $x_{e'}$ on all other edges $e' \in T$, specifies that the parent and child node each receive an equal incremental benefit from participating in the transaction.


Let $C_{s}$ be the set of children of $s$ in $T$. Let $s'$ be the parent of $s$ (if $s = r$, we consider a fictitious parent $r'$ of $r$, with the revenue share on edge $(r, r')$ always set to $0$). Define $\wst = \max_{t'\in C_s\setminus t} w_{t'}x_{t's}$. This is the maximum value that would reach $s$ given a set of shares $x$ if $t$ did not exist as a child of $s$.
Then, the two-player egalitarian bargaining solution on $(t,s)$ specifies splitting $w_t$, the value reaching node $t$, according to $x_{ts}$ where $x_{ts} \in [0,1]$ satisfies
\begin{equation}
\label{eqn:edge}
(1-x_{ts})w_t = (1-x_{ss'})\left(\max(\wst, w_{t}x_{ts}) - \wst \right).
\end{equation}

The left-hand side is the incremental benefit to node $t$ from transacting with $s$: it receives a payoff of $(1-x_{ts})w_t$ if it retains the edge with $s$, and nothing if it cuts off the edge. The right-hand side is the incremental benefit to the parent node $s$: if it retains the edge $(t,s)$, $s$ can choose the highest payoff from $C_s$ of which it will keep a $(1-x_{ss'})$ share (since it needs to share this payoff with its parent); if it cuts off the edge $(t,s)$, it only gets $(1-x_{ss'})$ times the highest payoff from the set $C_s \setminus t$.

The system of bargaining equations is given by writing (\ref{eqn:edge}) for all edges in the tree. A solution to this system is a fixed point of the bargaining game on the tree. 

{\bf Note:} It may seem that Equation (\ref{eqn:edge}) implicitly assumes that the parent node $s$ indeed lies on the winning path because the payoff to $s$ is $(1-x_{ss'})(\max_{t \in C_s} x_{ts}w_t)$ {\em only if} $s$ lies on the winning path, and is $0$ otherwise.
However, we can show (see Lemma~\ref{lem:ones}) that when $x$ is {\em fixed point} of these equations as opposed to an arbitrary set of shares, and if $s$ does not lie on the winning path, then $x_{ss'} = 1$, so that the right hand side is indeed $s$'s payoff (viz., $0$). Thus, Equation (\ref{eqn:edge}) holds for a fixed point solution irrespective of whether or not $s$ lies in the winning path.

\section{Reduction to path bargaining}
\label{sec:reduction}
The fixed point computation on the tree can be reduced to finding a fixed point of bargaining equations on a {\em single path}-- the path from the least common ancestor of the highest value leaves to the root (if there is a unique leaf with highest value, this is the path from that leaf to the root). For want of space we omit this reduction.

We summarize the reduction as follows. Let $v^\star = \max_l v_l$ be the maximum value in $T$. Find the least common ancestor $s_0$ of the leaves $\{l_1, \ldots, l_k\}$ with
$v_{l_i} = v^*$.
Remove the entire subtree rooted at $s_0$, and replace it with a fictitious buyer with value $d_0 = v^*$ at $l^*=s_0$.

Let the path from $l^*$ to the root be of length $n$; call this path $P^*$. We relabel nodes from
$l^*$ to the root $0, 1, \ldots, n$ (so that $l^*$ is $0$ and the root is $n$).
For $i \in [n]$, $e_i$ is the edge connecting $i-1$ to $i$.  We can show (see Appendix~\ref{app:reduction}) that $x_e =
1$ for all other edges $e \in T$. So to each node $i = 1, \ldots, n$, we can add a single edge with $x_e = 1$, to a fictitious buyer--- this fictitious buyer's value is the largest value excluding $v^*$ in the subtree rooted at $i$.
Call this value $d_i$; this is node $i$'s disagreement
point, and we may also think of $d_i$ as node $i$'s bid for the item being sold. We refer to this reduced instance as a path because the only edges with unknown revenue shares $x_i$ lie on a path.
Denote this new path bargaining instance by $(P^*, \vec{d})$. Note that $d_i$ is {\em strictly} less than $d_0$ for $i = 1, \ldots, n$. The following theorem, proved in Appendix~\ref{app:reduction}, summarizes this reduction:
\begin{theorem}
\label{t-tree-to-path} Given an instance $(T,\vec{v})$ of the tree bargaining
problem, construct the path bargaining instance $(P^*,\vec{d})$ as described
above. Then, $\vec{x}$ is a fixed point for $T$ if and only if $x_e = 1$ for $e
\notin P^*$, and the shares $x_e$, $e \in p^*$ constitute a fixed point to the
path bargaining problem $(P^*, \vec{d})$.
\end{theorem}

\section{Existence and Uniqueness of Fixed Point}
\label{sec:fixedpoint}

We now investigate fixed points of the path bargaining problem, having shown
that every tree bargaining instance can be reduced to a path bargaining
instance. To maintain the flow for easier reading, all proofs in this section are deferred to Appendix~\ref{app:omit}.

Recall that the value at node $0$ is $d_0$ and the remaining values at the
leaves $d_1, d_2, \ldots, d_n$ are all strictly less than $d_0$. The share on edge $e_i$
is $x_i$. For notational convenience, we assume there is a fictitious edge
$e_{n+1}$ going up from the root to a fictitious node labeled $n+1$ with share
$x_{n+1} := 0$. For $i = 0, 1, \ldots, n$, define $w_i = d_0\prod_{j=1}^i x_j$,
i.e. the value that reaches node $i$. 

A fixed point solution $\vec{x} = \langle x_1, x_2, \ldots x_n \rangle$
satisfies the bargaining equations (\ref{eq:balance}) for all edges $i$, with
$x_i \in [0, 1]$: that is, it simultaneously solves the following system of equations,
one for each edge $e_i$:
\begin{equation} \label{eq:balance}
(1 - x_i)w_{i-1}\ =\ (1 - x_{i+1})(x_iw_{i-1} - d_i).
\end{equation}

We note that in replacing the $\max\{x_iw_{i-1}, d_i\}$ term by $x_iw_{i-1}$ on
the right-hand side of the bargaining equation, we have used the fact (shown in the proof of Lemma~\ref{lem:maxwin}) that we must have $w_{i-1} x_i \geq d_i$
in {\em any} fixed point $x_i$ since $d_i < d_0$.

We can rewrite each bargaining equation in two ways: the ``upward equation''
gives $x_{i+1}$ in terms of $x_i$:
\begin{equation} \label{eq:upward}
x_{i+1} = 1 - \frac{(1 - x_i)w_{i-1}}{x_iw_{i-1} - d_i} = 1 - \frac{w_{i-1} - w_i}{w_i - d_i}.
\end{equation}
The ``downward equation'' gives $x_i$ in terms of $x_{i+1}$:
\begin{equation} \label{eq:downward}
x_i = \frac{w_{i-1} + (1 - x_{i+1})d_i}{(2 - x_{i+1})w_{i-1}}.
\end{equation}

Now we show that a fixed point to the path bargaining equations always
exists, and is unique. The existence proof is via Brouwer's fixed point
theorem. We show that the mapping $f$ that is (essentially) obtained by
simultaneous updates to the shares on all edges using the downward
equations~(\ref{eq:downward}) is a continuous mapping from $[0, 1]^n$ to
itself. The uniqueness proof requires more effort. We write two equations for
$x_n$ in terms of $x_1$: one by using the upward equations~(\ref{eq:upward})
and one by using the downward equations~(\ref{eq:downward}). These equations
can be represented by two curves, and any intersection point of the two curves
leads to a fixed point. We next show that in the feasible range for the curves, one is strictly increasing, and the other strictly decreasing; thus there is a
unique intersection point. We now formalize this.

First, we use the upward equations to write $x_2, x_3, \ldots, x_n$ in terms of $x_1$ and $\vec{d}$. However, not every value of $x_1 \in [0, 1]$ will give us
to values of $x_i$ in $[0, 1]$ and $w_i > d_i$. We will say that $x_1$ is {\em
feasible} if it does lead to $x_i \in [0, 1]$ and $w_i > d_i$. The
following lemma characterizes some monotonicity properties of the $x_i$'s and $w_i$'s when written in terms of $x_1$.

\begin{lemma} \label{lem:increasing}
If $x_1' < 1$ is feasible, then for all $x_1 \in [x_1', 1)$, and for all $i =
1, 2, \ldots, n$:
\begin{enumerate*}
    \item $x_i \in [x_i', 1)$
    \item $w_i > d_i$.
    \item $\frac{dx_i}{dx_1} > 0$ (so $x_i$ is strictly increasing as a function of $x_1$).
    \item $\frac{dw_i}{dx_1} > 0$ (so $w_i$ is strictly increasing as a function of $x_1$).
\end{enumerate*}
Here, $x_i$, $w_i$ (resp. $x_i'$, $w_i'$) etc. are defined by
$x_1$ (resp. $x_1'$) using the upward equations (\ref{eq:upward}).
\end{lemma}

Since $x_i = 1$ for all $i$ is a feasible solution, in particular $x_1 = 1$ is feasible, and Lemma~\ref{lem:increasing} immediately implies the following structure of the feasible region:
\begin{lemma} \label{lem:feasible-region}
Let $x_1^\circ = \inf\{x_1:\ x_1 \text{ is feasible}\}.$ Then the
feasible region for $x_1$ is either the interval $[x_1^\circ, 1]$ or $(x_1^\circ, 1]$, depending on whether $x_1^\circ$ is feasible or not.
\end{lemma}

If $x_1$ is feasible, and $x_2, \ldots, x_n$, are computed using the upward equations, then the balance conditions for edges $e_1, e_2, \ldots, e_{n-1}$ are automatically satisfied. The equation for $e_n$ may not be satisfied, however. A fixed point is obtained precisely when $x_n$ satisfies the balance condition for $e_n$. Geometrically, equations (\ref{eq:upward}) and (\ref{eq:downward}) for $x_n$ define two curves, the upward curve, and the downward curve respectively. A fixed point is obtained at any intersection point of the two curves for $x_n$ in the feasible region of $x_1$. The following lemma gives monotonicity properties of the two curves:
\begin{lemma} \label{lem:curves}
In the feasible region for $x_1$, the upward curve for $x_n$ is strictly increasing, and the downward curve for $x_n$ is strictly decreasing.
\end{lemma}
We immediately get our uniqueness result:
\begin{theorem}[Uniqueness] \label{lem:uniqueness}
If a fixed point to the equations (\ref{eq:balance}) exists, then it is unique.
\end{theorem}
\begin{proof}
This is immediate from Lemma~\ref{lem:curves}: a strictly increasing and
strictly decreasing curve can intersect in at most 1 point.
\end{proof}

Finally, using Brouwer's fixed point theorem we can show that a fixed point always exists:
\begin{theorem}[Existence] \label{lem:existence}
A fixed point to the bargaining equations (\ref{eq:balance}) exists.
\end{theorem}
Briefly, we consider the following function $f: [0, 1]^n \rightarrow [0, 1]^n$, which
represents a simultaneous update of the shares vector $\vec{x}$ on all edges
using the downward equations:
$$f_i(x) = \min\left\{\frac{w_{i-1} + (1 - x_{i+1})d_i}{(2 - x_{i+1})w_{i-1}},\ 1\right\},$$
where $w_{i-1} = d_0\prod_{j=1}^{i-1} x_j$ as usual, $x_{n+1} := 0$, and we make
the convention that when $w_{i-1} = 0$, the first expression in the
minimum above is $+\infty$, so that $f_i(x) = 1$. The above function is continuous,
and its domain $[0, 1]^n$ is a convex, compact set. By Brouwer's Fixed Point
Theorem, $f$ has a fixed point. The main work in the proof of Theorem~\ref{lem:existence} then consists in showing that any fixed point of $f$ is a fixed point to the bargaining equations (\ref{eq:balance}).

\section{Properties of the Fixed Point}
\label{sec:prop}


\subsection{Core Property}
The setting we study is naturally modeled as a cooperative game $(T,V)$, where the agents are the nodes in the trading tree $T$, and the coalition values $V$ are defined as follows. The value of the coalition consisting of nodes on a path $p = (l, i_1, \ldots, i_k, r)$ is $V(p) = v_l$. A coalition cannot generate value unless it contains a path from a leaf to the root; if it does contain such paths, its value is the maximum value amongst these paths: $V(S) = \max_{p \in S} V(p)$, and $V(S) = 0$ if $S$ does not contain any such path $p$. Note specifically that $V(S) = 0$ for all sets that do not contain the seller $r$, and that $V(T) = V(p^*) = v^*$.

The {\em core}~\cite{Leytonbook} of a cooperative game $(N,V)$ is defined as a set of nonnegative payoff vectors $(u_1, \ldots, u_N)$ with $\sum u_i = V(N)$ such that every coalition's total payoff is at least as much as the value it generates: $\sum_{i \in N} u_i \geq V(S)\quad \forall S$.
The core consists of the set of payoff vectors that are {\em not blocked} by any coalition which can increase its total payoff by splitting from the grand coalition and playing amongst themselves --- an outcome not in the core is unlikely to occur in practice since there is a coalition that can benefit by deviating. In general, the core of a game can be empty, but our particular cooperative game does have a non-empty core, and in fact, our fixed point lies in the core. We can show the following theorem (proved in Appendix~\ref{app:core}):
\begin{theorem} \label{thm:core}
The payoff vector $u^*$ belongs to the core of $(T,V)$.
\end{theorem}

\subsection{Monotonicity}

Monotonicity, which means that increasing the bargaining power of an agent increases his payoff, is a desirable property for a solution concept to our game. We establish a {\em strict} monotonicity property for the payoff to all nodes on the winning path in terms of their bargaining power. Since we are only interested in nodes on the winning path, we can restrict ourselves to discussing reduced path instance $P^*$.\footnote{When there is more than one leaf with value $v^*$, $P^*$ does not contain all nodes on the winning path, but the strict monotonicity result extends easily to that case since an increase in bargaining power for a winning node not in $P^*$ means that there is now a leaf with value greater than $v^*$.} (We note that strict monotonicity cannot hold for nodes outside the winning path since the outcome itself must change for these nodes to receive a nonzero payoff; however, a weak monotonicity condition trivially holds.) We can prove the following strict monotonicity property (proved in Appendix~\ref{app:monotonicity}):
\begin{theorem} \label{thm:monotonicity}
Consider the path bargaining problem. If any $d_i$ is increased (but is still kept less than $d_0$) while the remaining $d_{i'}$ are unchanged, the payoff of $i$
strictly increases.
\end{theorem}

\subsection{Computability}
\label{sec:compute}

We know that there exists a unique fixed point, $x^\star$, of the bargaining equations. We now turn to computability of the fixed point. Note that since the shares affect the bids multiplicatively, the fixed point solution is {\em
scale-free}: if we scale all bids by the same amount, the fixed point stays the same. So to simplify calculations, we assume that the maximum bid, $d_0$, is
normalized to $1$, and all other bids $d_i$ are less than $1$. We can give a
polynomial-time algorithm to compute an $\epsilon$-fixed point: i.e., a set of
shares such that all bargaining equations are satisfied within an additive
$\epsilon$ error. For the original unscaled bids where the maximum
bid may not be equal to $1$, the additive error gets scaled by the maximum bid
as well.

We now state our theorem (proved in Appendix~\ref{app:computability}) regarding computability of an approximate fixed point. It is given in terms of a parameter $\gamma = \min\{1 - \max_{i > 0} \{d_i\}, \frac{1}{n+2}\}$, which is essentially how close the second highest bid is to the maximum. Note that the dependence on the error parameter $\epsilon$ and $\gamma$ is only poly-logarithmic. In practice, the algorithm converges extremely fast.

\begin{theorem} \label{thm:eps-fixed point}
There is an algorithm that, for any given $\epsilon > 0$, computes an $\epsilon$-fixed point to the bargaining equations (\ref{eq:balance}) in $\text{poly}(n, \log(1/\gamma\epsilon))$ time.
\end{theorem}

The algorithm essentially works by running a binary search to find the intersection point of the upward and downward curves for $x_n$. The parameter $\gamma$ is important in giving bounds on the number of iterations needed in the binary search to obtain the desired accuracy, essentially by obtaining quantitative versions of the arguments of Section \ref{sec:fixedpoint}. 


\subsection{Dynamics}
\label{sec:dynamics}

We have already seen in the previous sections that the solution prescribed by the fixed point of the bargaining equations has several desirable properties. A natural question is whether the agents on the tree would, without help from a centralized authority, be able to converge to this fixed point. We now present numerical evidence that this is indeed likely. 
Our experiments suggest that a natural dynamics consisting of {\em asynchronous} updates--- where in each step a random edge $e$ updates $x_e$ according to the two-player egalitarian bargaining equation (\ref{eqn:edge}), using the current values of $x_{e'}$ on other edges--- indeed converges to the fixed point.

We run $10,000$ tries of the following experiment: generate random bids at the leaves of a depth-$8$ balanced binary tree with $256$ leaves and $510$ edges; this is a convenient size that permits $10,000$ tries to be run in a few hours. The bids are drawn from the lognormal distribution $e^{(1+N)}$ where $N$ is the normal distribution with zero mean and unit variance. We initialize all $510$ edge multipliers to the arbitrary value $0.99$, and then repeatedly re-negotiate the edge multipliers one at a time in a random order: the negotiation for each edge consists of solving equation (\ref{eq:balance}) for that edge (while freezing the values of all other
multipliers). More specifically, binary search down to an tolerance of $1.0 \times 10^{-15}$ is used to solve the equation. The edge updates are organized into ``rounds'' during each of which every edge is individually updated in a random order specified by a different random permutation for every round.

We continue iterating until the solution is close enough to the fixed point
computed using the reduction to the path and the algorithm in Section \ref{sec:compute}. The efficient fixed point finding algorithm uses the reduction of Section \ref{sec:reduction} to convert the tree problem to a path problem. This path problem is then solved using a heuristic program (not described here) that uses the algorithm of Section \ref{sec:compute} as a subroutine and computes the $8$ multipliers to a nominal accuracy of $2.0 \times 10^{-16}$. The multipliers for the original tree are obtained by copying those $8$ values onto the winning path and then setting the $502$ multipliers lying on side branches to the value $1.0$.

{\em Every one} of the 10,000 tries converged to the desired tolerance. The plot in Figure \ref{fig:exp} shows the average convergence rate summarizing those 10,000 tries. It is clear that the shares always converge to within the desired accuracy at a reasonable rate. While we do not include the figures here, we also observed similar convergence behavior on trees with different structures and sizes, as well as for several different bid distributions.

\begin{figure}[ht]
\centering
\includegraphics[width=3in]{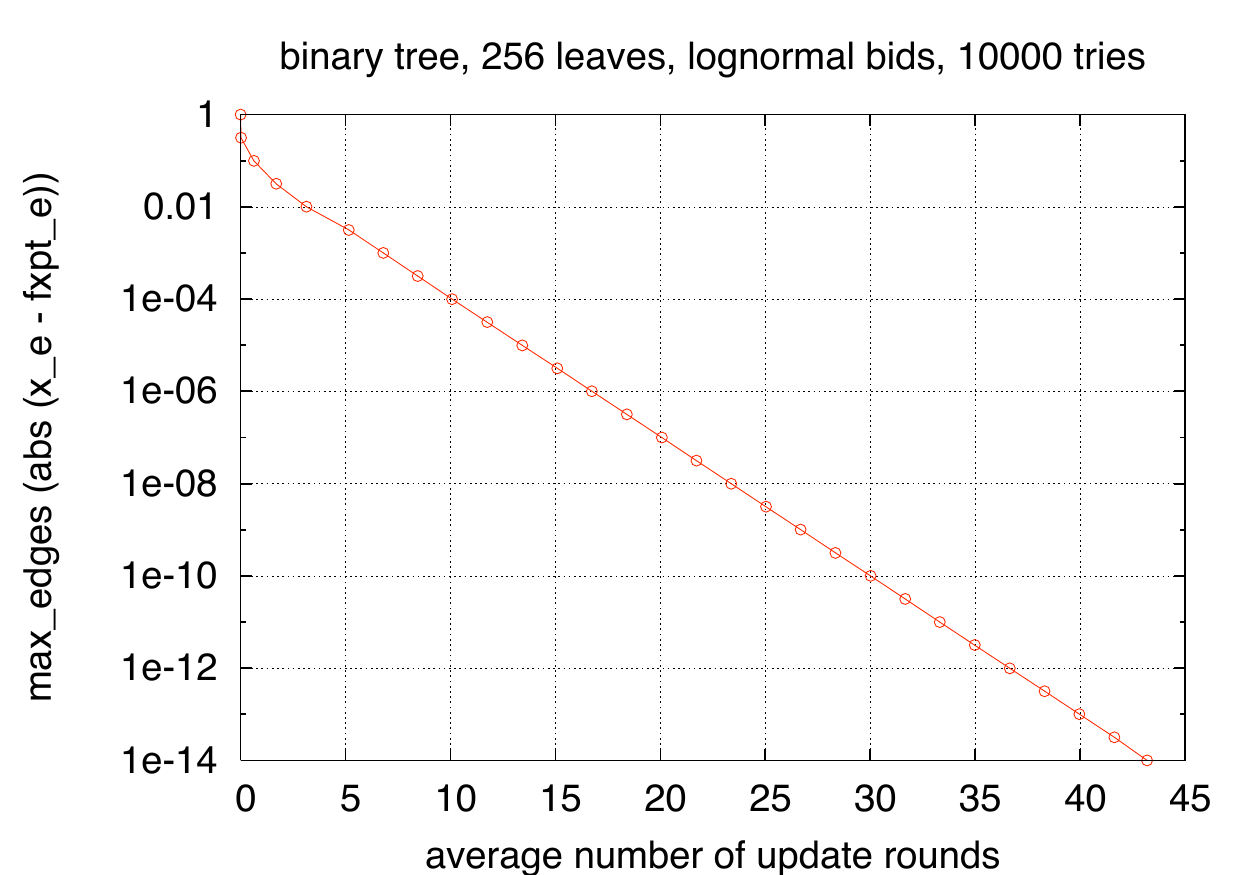}
\caption{Asynchronous dynamics convergence: accuracy vs. average number of rounds to achieve accuracy.}\label{fig:exp}
\end{figure}



\section{Further Directions}
In this paper, we defined a bargaining game on trees motivated by a fair division question in display ad exchanges, and investigated the properties of its fixed point. There are a number of interesting directions for further work.  The most interesting open question is proving the convergence of dynamics, since numerical simulations strongly suggest that even asynchronous dynamics converge to the fixed point. Another interesting direction is that of a Bayesian model for values--- suppose instead of values $v_i$ at the leaves, we had {\em distributions} of values. The problem of solving the bargaining equations to set the shares $x_e$ in this case is a very meaningful one, but also one that appears to be technically extremely challenging. 
Finally, there are questions related to extending the trade network model itself: for example, in this paper, we only consider a single seller and a tree topology. The question of how to model and solve for multiple sellers, and how the fixed point behaves if the underlying trade network is a directed acyclic graph instead of a tree, are also interesting directions for further work.

\paragraph{Acknowledgments.} We are very grateful to Nikhil Devanur and Mohammad Mahdian for insightful discussions on network bargaining and cooperative games, and to Matt Jackson, Preston McAfee, Herve Moulin, Michael Schwarz, and Mukund Sundararajan for helpful comments and pointers to relevant literature.

\bibliographystyle{alpha}
\bibliography{bargain}

\newcommand{\etalchar}[1]{$^{#1}$}
\begin{thebibliography}{FMMP10}

\bibitem[ABC{\etalchar{+}}09]{AzarBCDP09}
Y.~Azar, B.~E. Birnbaum, L.~E. Celis, N.~R. Devanur, and Yuval Peres.
\newblock Convergence of local dynamics to balanced outcomes in exchange
  networks.
\newblock In {\em FOCS}, pages 293--302, 2009.

\bibitem[ADJR10]{AzarDJR10}
Y.~Azar, N.~R. Devanur, K.~Jain, and Y.~Rabani.
\newblock Monotonicity in bargaining networks.
\newblock In {\em SODA}, pages 817--826, 2010.

\bibitem[BEKT07]{BlumeEKT07}
L.~Blume, D.~A. Easley, J.~M. Kleinberg, and {\'E}.~Tardos.
\newblock Trading networks with price-setting agents.
\newblock In {\em EC}, pages 143--151, 2007.

\bibitem[CDP10]{CelisDP10}
L.~E. Celis, N.~R. Devanur, and Y.~Peres.
\newblock Local dynamics in bargaining networks via random-turn games.
\newblock In {\em WINE}, pages 133--144, 2010.

\bibitem[CK08]{ChakrabortyK08}
T.~Chakraborty and M.~Kearns.
\newblock Bargaining solutions in a social network.
\newblock In {\em WINE}, pages 548--555, 2008.

\bibitem[CKK09]{ChakrabortyKK09}
T.~Chakraborty, M.~Kearns, and S.~Khanna.
\newblock Network bargaining: algorithms and structural results.
\newblock In {\em EC}, pages 159--168, 2009.

\bibitem[FMMP10]{FeldmanMMP10}
J.~Feldman, V.~S. Mirrokni, S.~Muthukrishnan, and M.~M. Pai.
\newblock Auctions with intermediaries: extended abstract.
\newblock In {\em EC}, pages 23--32, 2010.

\bibitem[Kal77]{kalai}
E.~Kalai.
\newblock Proportional solutions to bargaining situations: Interpersonal
  utility comparisons.
\newblock {\em Econometrica}, 45(7):1623--30, October 1977.

\bibitem[Kan10]{Kanoria10}
Y.~Kanoria.
\newblock An {FPTAS} for bargaining networks with unequal bargaining powers.
\newblock In {\em WINE}, pages 282--293, 2010.

\bibitem[KT08]{KleinbergT08}
J.~M. Kleinberg and {\'E}.~Tardos.
\newblock Balanced outcomes in social exchange networks.
\newblock In {\em STOC}, pages 295--304, 2008.

\bibitem[LBS08]{Leytonbook}
K.~Leyton-Brown and Y.~Shoham.
\newblock {\em Essentials of Game Theory: A Concise, Multidisciplinary
  Introduction}.
\newblock Morgan and Claypool Publishers, 2008.

\bibitem[Mou04]{moulin}
H.~Moulin.
\newblock {\em {Fair Division and Collective Welfare}}.
\newblock The MIT Press, September 2004.

\bibitem[Yah07]{YahooRM}
http://advertising.yahoo.com/article/right-media.html, 2007.

\end{thebibliography}

\appendix
\section{Example of bargaining equations}
\label{app:example}

\begin{figure}[ht]
\centering
\includegraphics[width=1.8in]{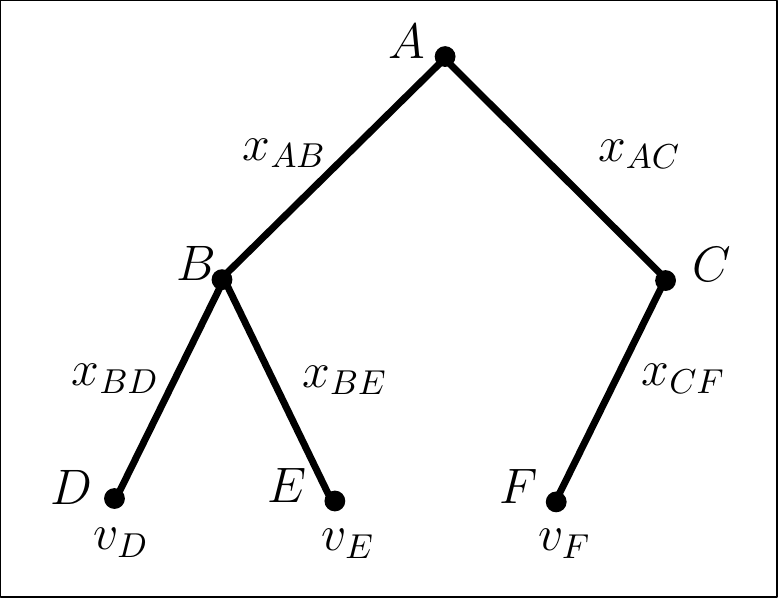}
\caption{A trading tree.}\label{fig:tree}\vspace{-.5cm}
\end{figure}

\begin{align*}
&\text{Equation for edge} (D, B):\\
&(1-x_{DB})v_{D}\ =\ (1-x_{BA}) ( \max \{ x_{DB} v_D,  x_{EB} v_E \} - x_{EB} v_E) \\
&\text{Equation for edge} (E, B):\\
&(1-x_{EB})v_{E}\ =\ (1-x_{BA}) ( \max \{ x_{DB} v_D, x_{EB} v_E \} - x_{DB} v_D) \\
&\text{Equation for edge} (F, C):\\
&(1-x_{FC})v_{F}\ =\ (1-x_{CA}) (  x_{FC} v_F  )  \\
&\text{Equation for edge} (B, A):\\
&(1- x_{BA})(\max\{ x_{DB} v_D, x_{EB}  v_E \})\ =\ \max\{ x_{BA}\max\{ x_{DB} v_D, x_{EB}  v_E \}, x_{CA} x_{FC}  v_F \} - x_{CA} x_{FC}  v_F \\
&\text{Equation for edge} (C, A):\\ 
&(1- x_{CA}) x_{FC}  v_F\ =\ \max\{x_{BA}\max\{ x_{DB} v_D, x_{EB}  v_E \}, x_{CA} x_{FC}  v_F \} - x_{BA}\max\{x_{DB} v_D, x_{EB} v_E\}  
\end{align*}

\section{Reduction to Path Instance}
\label{app:reduction}

We will assume henceforth that $T$ has been pruned to remove all buyers with value zero, \ie, $v_l > 0$ for all leaves $l$ in $T$. Recall also that when we write $e = (t,s)$, $t$ is the child and $s$ the parent.

Our first lemma shows that the revenue shares in any fixed point must be nonzero, that is, $x_e \in (0,1]$ for any fixed point $\vec{x}$.

\begin{lemma}
\label{lem:nonzero}
Let $\vec{x}$ be any fixed point to the bargaining equations on $T$. For any edge $e\in T$, $x_e >0$.
\end{lemma}
\begin{proof}
Suppose not. Choose edge $e=(t,s)$ with $x_{ts} = 0$ such that every edge in the subtree rooted at the child node $t$ has share $x_{e'} > 0$. Since $x_{ts} = 0$, the marginal benefit to the parent node $s$ from the edge $(t,s)$ is zero, \ie, $\max(w_tx_{ts}, \wst)  -  \wst= 0$.

Since $v_l > 0$ for all $l \in T$ and $x_{e'} > 0$ for all $e'$ in the subtree rooted at $t$, we must have $w_t > 0$. But then $(1-x_{ts})w_t = w_t > 0$, a contradiction to $\vec{x}$ satisfying the bargaining equation on edge $e=(t,s)$.
\end{proof}

Our second lemma allows us to show that if the share $x_e$  on an edge $e$ is $1$ in a fixed point (\ie, the child passes up all the value to the parent), the share on {\em every} edge in the subtree below $e$ must be $1$ as well in that fixed point.
\begin{lemma}
Consider edges $(t,s),(s,s')$ such that $s$ is the parent of $t$ and $s'$ is the parent of $s$, and let $\vec{x}$ be any fixed point.  If $x_{ss'} = 1$ then $x_{ts} =1$ also.
\end{lemma}

\begin{proof}
By definition, the fixed point satisfies the bargaining equations on edge $(t,s)$:
\[
(1-x_{ts})w_t = (1-x_{ss'})(\max(\wst,x_{ts}w_t) - \wst).
\]
Since $x_{ss'} = 1$, the right-hand side is zero. Also, $w_t >0$, since $x_e >0$ on every edge $e$ in a fixed point by Lemma \ref{lem:nonzero} and $v_l > 0$ for all leaves $l$.  Therefore, to satisfy  $(1-x_{ts})w_t = 0$, we must have $x_{ts} = 1$.
\end{proof}

The following corollary follows immediately.
\begin{corollary}
\label{cor:onessubtree}
Consider any child node $t$ with parent $s$ in $T$. If $x_{ts} =1$ in a fixed point, then $x_{e} = 1$ for every edge $e$ in the subtree rooted at $t$ in this fixed point.
\end{corollary}

Next we show that the path corresponding to a buyer with the highest value will also offer the highest final value to the seller, and therefore win: that is, the outcome corresponding to any fixed point is always efficient.

\begin{lemma}
\label{lem:maxwin}
Let $v^*= \max_{l \in T} v_l$ be the highest value on the leaves, and let $l^*$ be a highest value leaf with corresponding path $p^*$ to the root.  Let $l$ be any leaf with value $v_l < v^*$ and corresponding path $p$ to the root. Then, in any fixed point $\vec{x}$,
\[
v^* \left( \prod_{e \in p^*} x_e \right) > v \left( \prod_{e \in p} x_e \right).
\]
\end{lemma}
\begin{proof}
Suppose not. Then, there exists a path $p = \{l, i, \ldots, r\}$ with $v_l < v^*$ and $$v^* \left( \prod_{e \in p^*} x_e \right) \leq v \left( \prod_{e \in p} x_e \right).$$

Let $s$ be the least common ancestor of $l$ and $l^*$ in $T$. Let $t$ be the child of $s$ such that $(t,s) \in p$ and let $t^*$ be the child of $s$ such that $(t^*,s) \in p^*$. Let $p^*_{l^*s}$ be the portion of path $p^*$ from $l^*$ to $s$, and similarly $p_{ls}$ be the portion of path $p$ from $l$ to $s$. Let $s'$ be the parent of $s$, and consider the bargaining equation for the edge $(s,t^*)$,
\[
(1-x_{st^*})w_{t^*} = (1-x_{ss'})(\max(w_{s\setminus t^*}, x_{st^*}w_{t^*}) - w_{s\setminus t^*}).
\]
Since $v^* \left( \prod_{e \in p^*} x_e \right) \leq v_l \left( \prod_{e \in p} x_e \right)$, and the paths $p^*$ and $p$ differ only 'below' $s$ (by choice of $s$), it must be that  $v^* \left( \prod_{e \in p^*_{l^*s}} x_e \right) \leq v_l \left( \prod_{e \in p_{ls}} x_e \right)$ as well. Therefore, the value offered to $s$ by $t^*$ is no larger than the value offered by $t$, or $w_{s\setminus t*} \geq x_{st^*}w_{t^*}$. Therefore, the right hand side of the bargaining equation above is zero.

As before, $w_{t^*} > 0$ in any fixed point from Lemma~\ref{lem:nonzero} and because $v_l > 0$ for all leaves $l$.  Therefore, it must be the case that $x_{st^*} = 1$.  But then $x_e = 1$ for all $e \in p^*_{l^*s}$ by Corollary~\ref{cor:onessubtree}, which says that all edges 'below' $t^*$ must have a share of $1$ also. But then $v^* \left( \prod_{e \in p^*_{l^*s}} x_e \right) = v^* > v \geq v\left( \prod_{e \in p_{ls}} x_e \right)$ since $v < v^*$ and $x_e \leq 1$, a contradiction.
\end{proof}
The following corollary follows immediately from the result above, since every parent node chooses the highest offer from amongst its children.
\begin{corollary}
The outcome corresponding to any fixed point $\vec{x}$ is efficient, \ie, the winner is a leaf with highest value.
\end{corollary}

Recall that our goal is to investigate fixed points $\vec{x}$ of the system of bargaining equations, one bargaining equation for each edge $e$ in $T$. Next, we identify the value of $x_e$ in any fixed point on all but one (sub)path in the tree. Our first lemma tells us that the values of $x_e$ for all edges $e \notin p^*$, where $p^*$ is a path from a highest value buyer to the root, must be $x_e = 1$ in {\em any} fixed point, if at all one exists.

\begin{lemma}
\label{lem:ones}
Let $e$ be any edge that is not on a path from a highest value leaf to the root. Then, in {\em any} fixed point $\vec{x}$, $x_e =1$.
\end{lemma}
\begin{proof}
Let $p^*$ denote the path from a highest value leaf to the root. Consider an edge $e$ on a path from a leaf $l$ with value $v_l < v^*$ to the root, with $x_{e} < 1$. Let $s$ be the least common ancestor of $l^*$ and the child endpoint of $e$. Let $t$ be the child of $s$ that is on the path from $s$ to the edge $e$ (note that $s$ could be the parent endpoint of $e$, in which case $t$ is the other (child) endpoint of $e$).

Let $s'$ be the parent of $s$, and consider the bargaining equation for edge $ts$,
\[
(1-x_{ts})w_{t} = (1-x_{ss'})(\max(w_{s\setminus t}, x_{ts}w_{t}) - w_{s\setminus t}).
\]

Since $v_l< v^*$, we have from the proof of Lemma~\ref{lem:maxwin} and by choice of $s$ that $\wst >  x_{ts}w_t$. Therefore, the right-hand side of the bargaining equation is $0$ and so $(1-x_{ts})w_{t} = 0$ as well; as before, this can only happen with $x_{ts} =1$. Since $e$ belongs to the subtree rooted at $s$,  by Corollary \ref{cor:onessubtree}, $x_e = 1$ as well.
\end{proof}

In general, there could be more than one such path $p^*$ if there is more than one leaf with value $v_l = v^*$; so far, we do not know how to deal with these multiple paths. Our next lemma shows us how to deal with this. The proof is very similar to the previous lemmas, and is in Appendix C.
\begin{lemma}
\label{l-ties}
Suppose there is more than one leaf with $v_l = v^*$. Let $\{l_1, \ldots, l_k\}$  be leaves with $v_{l_i} = v^*$, and let $s$ be the least common ancestor of $\{l_1, \ldots, l_k\}$ in $T$. Then, in any fixed point, $x_e = 1$ for all edges in the subtree rooted at $s$.
\end{lemma}
In particular, this lemma tells us that if there are two buyers with maximum value in different subtrees rooted at the seller, the seller extracts full value.

\begin{proof}
Let $(t_i,s)$ be the child of $s$ on path $p_{l_is}$ from leaf  $l_i$ to $s$, and let $s'$ be the parent of $s$.
First, recall from the proof of Lemma \ref{lem:maxwin} that at every internal node $i$, $v^*\prod_{e \in p_{l_j i}	} x_e > v\prod_{e \in p_{li}} x_e$ for all leaves  $l$ with value $v< v^*$ in any fixed point. Therefore, the maximum value at every internal node must come from one of the paths corresponding to $\{l_1, \ldots, l_k\}$.
So the bargaining equation for the edges $(t_i,s)$ can be written as
\[
(1-x_{st_i})w_{t} = (1-x_{ss'})(\max_{j\in\{1, \ldots, k\}}(x_{st_j}w_{t_j}) -\max_{j\neq i}(x_{st_j}w_{t_j})).
\]
Let $i^* = \arg\max_{j}(x_{st_j}w_{t_j})$. Then, for any $i \neq i^*$, the right-hand side of the bargaining equation above is zero (note that this is true even if the $\arg\max$ is not unique). As before, $w_{t_i} > 0$ in the right-hand side, so we must have $x_{st_i} = 1$ for all $i \neq i^*$. By Corollary \ref{cor:onessubtree}, this means that $x_e = 1$ for all edges in the subpath from $s$ to $l_i$. Therefore $w_{t_i} = v^*$ and $x_{st_i}w_{t_i} = v^*$ as well. But this is the maximum possible value for $x_{st_j}w_{t_j}$ since all the $x_e \leq 1$. Therefore, the right-hand side of the bargaining equation is zero for $i = i^*$ as well, and $x_{st_{i^*}} =1$ also; again by Corollary \ref{cor:onessubtree}, $x_e= 1$ for the remaining edges in this subpath as well. Since we already know that $x_e = 1$ for edges on all other paths in $T$, this proves our claim.
\end{proof}

Together, these lemmas will be adequate to reduce the tree bargaining problem to path bargaining problem.

We make a brief aside before summarizing the reduction to a path from a tree. Recall that given a set of revenue shares $x$, we defined nodes' payoffs in our model by computing the winner recursively at every parent node until the root, and then setting the payoffs of nodes not on the winning path to $0$, and the payoffs of nodes on the winning path to be those given by $x$. The following proposition states that the payoffs to nodes when the revenue shares $x$ constitute a {\em fixed point} of the bargaining equations can be written {\em directly} in terms of the $x$; that is, we do not need to `manually' set the payoffs to zero for nodes off the winning path.
\begin{proposition}
\label{prop:payoff}
Suppose $\vec{x}$ is a fixed point of the bargaining equations. Then each node's payoff can be computed recursively as $u_t = (1-x_{ts})w_t$.
\end{proposition}
The easy proof follows from Lemmas \ref{lem:ones} and \ref{l-ties} which tell us that $x_e = 1$ for all edges not on a path from a highest value leaf to the root; therefore $(1-x_{ts})w_t= 0$ corresponding to the zero payoff received by these nodes in our model.

Using these lemmas, the reduction described in Section \ref{sec:reduction} follows. An example of this reduction is shown in Figure~\ref{fig:reduce}.

\begin{figure}[ht]
\centering
\includegraphics[width=3in]{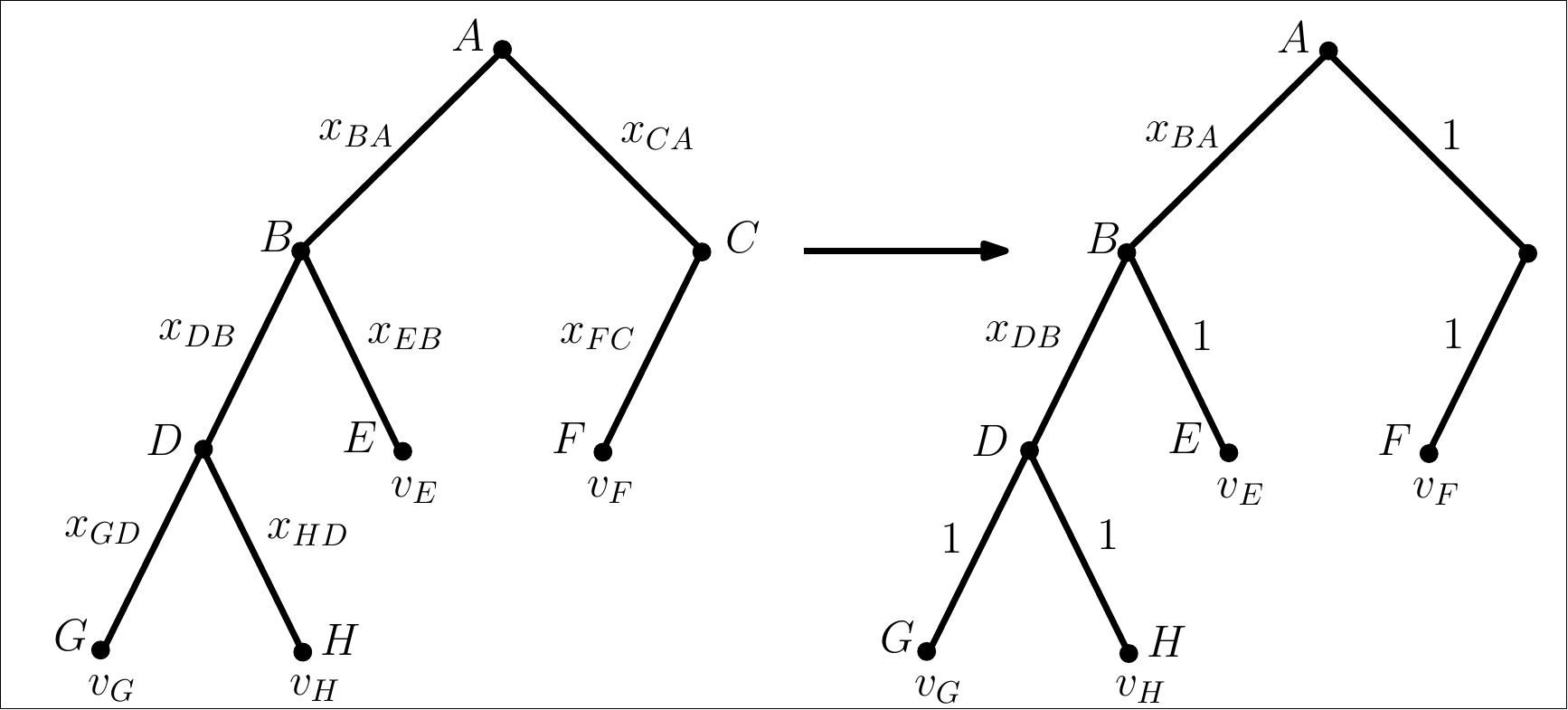}
\caption{Reduction from tree to path.  The largest value is $v_G = v_H$.}\label{fig:reduce}\vspace{-.5cm}
\end{figure}

\section{Other Solution Concepts}
\label{app:other}
Here we discuss other possible ways to choose a winner and distribute the payoff amongst nodes in the tree.  While we use a bargaining game to determine the outcome many other allocations might seem natural as well. For instance, one possible outcome is to choose the winner as the leaf with a highest value, and allocate this maximum amongst the nodes on $p^*$ by starting at the leaf, and assigning to each node the incremental value it adds to the tree.  in the special case with no intermediaries and only buyers and the seller ,  this leads to the same outcome as a second price auction. There are different variants of the scheme, one of which we will detail next --- however, as the example shows, such allocations have the undesirable property of possibly assigning zero payoffs to nodes on the winning part, even though they are crucial to generating the value $v^*$.

Given an instance $(T,v)$,  compute an outcome as follows: the path to the highest buyer is chosen as the winning path and only nodes on this path receive nonzero payoffs.  For a node $s$ on the winning path, let $r(s)$ be the remaining unallocated value after assigning value to all nodes in the path from the leaf to $s$.  Let $\alpha(s)$ be the largest value at a buyer if the subtree rooted at $s$ is removed from the tree.  Then, the value given to $s$ is $r(s) - \alpha(s)$.  For an example, consider Figure \ref{fig:tree} and let $v_D$ be the highest value at a buyer.  The value allocation would be $D: v_D - \max\{v_E, v_F\}$; $B:v_D - \max\{v_E, v_F\} -v_F$; $A: v_F$.   Intuitively, this allocation assigns each player $s$ the additional value she brings to the network through her subtree (\ie, the subtree rooted at $s$) {\em excluding} those nodes in the winning path.  Notice that if $v_F > v_E$ then $B$ receives zero value or, similarly, if $v_F = 0$ then the seller $A$ receives zero value.

A possible alternative model is to impose a fixed cost plus revenue share price structure--- each parent node charges a fixed cost to a child node for the connection that it provides the child node, and in addition keeps some share of the value that it passes up towards the seller. The question of how to choose such a fixed cost, and how the shares should be chosen given the fixed costs, is beyond the scope of this paper.

Next we discuss the two standard solution concepts in cooperative game theory, the Shapley value and the nucleolus. Recall that we prove that the fixed point leads to a payoff vector in the core of the cooperative game.

\subsection{Shapley value not in core}
The Shapley value for our cooperative game $(T,V)$ can award nonzero payoffs to nodes not in the winning path, which is clearly an undesirable property. The simplest example for this is when the tree has no intermediaries, and simply the seller and two buyers, with values, say, $12$ and $6$ respectively. In this case, the Shapley value of the game  awards a payoff of $7$ to the seller, and $4$ and $1$ to the two buyers respectively. It is easy to see that the example extends to larger trees with intermediaries, again giving nonzero payoffs to nodes not on $p^*$.

By Proposition \ref{prop:offpath}, such a payoff vector cannot belong to the core; therefore, the Shapley value need not belong to the core in our cooperative game.

\subsection{Nucleolus is not strictly monotone}
The nucleolus is a standard solution concept in cooperative game theory~\cite{Leytonbook}, and always belongs to the core when the core is nonempty. However, as we will see next, the nucleolus does not possess the strict monotonicity property possessed by the fixed point of the bargaining equations. This is analogous to the lack of strict monotonicity of the nucleolus in the network bargaining game as shown in~\cite{AzarDJR10} (we note though that the alternative solution concept we consider can be computed in polynomial time in contrast to the other solution concepts discussed in~\cite{AzarDJR10}).

The example below shows that the nucleolus need not be strictly monotone with $d_i$.

Consider a path with three nodes with $d_0 = 5$, $d_1 = 1$ and $d_2 = 3$. It can be verified that the nucleolus for the cooperative game $(T,V)$ corresponding to the tree awards payoffs $x_0 = 2/3$, $x_1 = 2/3$ and $x_2 = 11/3$ (recall that node $0$ is the buyer and node $2$ is the seller).

Now suppose $d_1$ increases to $d_1 = 2$, while all other values remain the same. Again, it can be verified that the nucleolus remains $x_0 = 2/3$, $x_1 = 2/3$ and $x_2 = 11/3$. That is, the payoff awarded to node $1$ by the nucleolus does not increase even though the bargaining power, quantified by $d_1$ increased. This is in contrast to the payoffs resulting from the bargaining solution, which are {\em strictly} monotone in the $d_i$.

\section{Nash bargaining}
\label{app:nash}

Suppose the players at the endpoints of each edge were to decide on a split using Nash bargaining instead of egalitarian bargaining, taking the shares on other edges as given, as before. As we show below, using Nash bargaining instead of egalitarian bargaining can actually lead to an inefficient solution--- the winner need not be the buyer with the highest value. The reason this happens is that these bargaining equations do not allow the child node to factor in the outcomes at nodes higher up in the tree: the value $w_t$ being bargained about is {\em realized only if} the parent node receives an adequately high portion of $w_t$ to beat out other competitors at its parent node. In our formulation of the bargaining equations where the players engage in egalitarian bargaining,  the child node equates its payoff with the {\em net} payoff of the parent, accounting for the fact that the parent must pass up some of the value from bargaining with this child. Thus, if the net payoff to the parent is zero, corresponding to losing, the child receives zero payoff too--- roughly speaking, a child with adequately high value will therefore pass up as much value as necessary (or possible) to its parent to ensure a nonzero payoff which the egalitarian bargaining equations. When we write the two player Nash bargaining equation however, the child only bargains `myopically' without accounting for the fact that it needs to actually be on the winning path to receive any payoff. That is, it bargains for a fair share {\em assuming} its parent will win, rather than for a fair share that still gives the parent node enough value to win higher up in the tree.

We note that essentially the same problem occurs when an additive, rather than multiplicative, form is used for the splits, since the amount that the parent must pass up is subtracted from both terms on the RHS, namely the parent's benefit with or without this child---  when this term does not influence the solution to the bargaining, inefficiencies can arise.

To be precise, consider the Nash bargaining solution for an edge $e= (t,s)$ is given by the following optimization,
\begin{eqnarray*}
\max & (u_1 - d_1)(u_2-d_2)\\
& (u_1, u_2) \in S.
\end{eqnarray*}
where $u_1 = (1-x_{ts})w_t$ is the value obtained by $t$, $d_1 = 0$ is $t$'s disagreement point, $u_2 = (1-x_{ss'})(\max(x_{ts}w_t, w_{s\setminus t})$ is the value obtained by $s$, and $d_2 = (1-x_{ss'})w_{s\setminus t}$ is $s$'s disagreement point, and the feasible set $S$ is parameterized by constraining the variables $x_{ts}$ and $x_{ss'}$ to lie in $[0, 1]$. Assume that $w_t > w_{s\setminus t}$ without loss of generality.

Note that the amount that the parent needs to pass up, $(1-x_{ss'})$, factors out of the objective and therefore the solution to the maximization, and does not affect the final split at all. Solving this maximization problem we see that the solution has $x_{ts} = \frac{1}{2} + \frac{w_{s\setminus t}}{2w_t}$. Thus, the child node $t$ keeps $(w_t- w_{s\setminus t})/2$, and passes up the remainder, namely $(w_t + w_{s\setminus t})/2$ to the parent node $s$.

Since the share on the edge above does not figure in the bargaining equation for this edge, the only set of interdependent equations are the ones corresponding to different children of the same parent. It is easy to check that by the same reasoning as before, the child with the highest $w_t$ will `win' at the parent in any fixed point, and the shares on edges with lower $w_t$ will be $1$. The value that the parent receives will be $(w_t+w_{t'})/2$, where $w_t$ and $w_{t'}$ are the highest and second-highest values at the child nodes respectively. This can lead to an inefficient outcome: for example, consider the trading tree in Figure~\ref{fig:tree} with leaf node $F$ removed (so that $C$ becomes a leaf node). The values at the leaves are $v_D = 1$, $v_E = 0.1$ and $v_C = 0.6$. By the reasoning above, the value reaching $B$ in a fixed point solution is $(1 + 0.1)/2 = 0.55$. Thus, $B$ cannot beat $C$ even if $B$ passes their entire value $0.55$ to the seller $A$, and hence the winner is $C$ with value $0.6$, rather than $D$, with value $1$. This is thus an inefficient solution.

\section{Proof of Lemmas from Section \ref{sec:fixedpoint}}
\label{app:omit}

We restate and prove Lemma~\ref{lem:increasing}.
\begin{lemma}
If $x_1' < 1$ is feasible, then for all $x_1 \in [x_1', 1)$, and for all $i =
1, 2, \ldots, n$:
\begin{enumerate*}
    \item $x_i \in [x_i', 1)$
    \item $w_i > d_i$.
    \item $\frac{dx_i}{dx_1} > 0$ (so $x_i$ is strictly increasing as a function of $x_1$).
    \item $\frac{dw_i}{dx_1} > 0$ (so $w_i$ is strictly increasing as a function of $x_1$).
\end{enumerate*}
Here, $x_i$, $w_i$ (resp. $x_i'$, $w_i'$) etc. are defined by
$x_1$ (resp. $x_1'$) using the upward equations (\ref{eq:upward}).
\end{lemma}
\begin{proof}[Proof of Lemma \ref{lem:increasing}]
We prove this by induction on $i$. For $i = 1$, parts 1 and 3 are trivially
true part 4 follows from $w_1 = d_0x_1$. Finally, for part 2, we have $w_1 =
d_0x_1 \geq d_0x_1' = w_1' > d_1$, since $x_1'$ is feasible.

Assume that the statement is true for some $i \geq 1$, now we show it for $i+1$. First, by the induction hypothesis, we have $x_i < 1$, and $w_i > d_i$, and so $x_{i+1} = 1 - \frac{(1 - x_i)w_{i-1}}{w_i - d_i} < 1$. This proves half of part 1. We now prove parts 3 and 4. We have

\begin{eqnarray*}
&&\hspace{-.4cm}\frac{dx_{i+1}}{dx_1}\ =\ \frac{\partial x_{i+1}}{\partial x_i}\cdot \frac{dx_{i}}{dx_1} + \frac{\partial x_{i+1}}{\partial w_{i-1}}\cdot \frac{dw_{i-1}}{dx_1}\ \\
&&\hspace{-.4cm}=\ \frac{w_{i-1}(w_{i-1} - d_i)}{(x_iw_{i-1} - d_i)^2} \cdot \frac{dx_{i}}{dx_1} +  \frac{(1 - x_i)d_i}{(x_iw_{i-1} - d_i)^2} \cdot \frac{dw_{i-1}}{dx_1}\ >\ 0,
\end{eqnarray*}

since, by the induction hypothesis, we have $x_i < 1$, $w_{i-1} \geq w_i > d_i$, $\frac{dx_{i}}{dx_1} > 0$ and $\frac{dw_{i-1}}{dx_1} > 0$. Similarly, we have
\begin{eqnarray*}
&&\frac{dw_{i+1}}{dx_1}\ =\ \frac{\partial w_{i+1}}{\partial x_{i+1}}\cdot \frac{dx_{i+1}}{dx_1} + \frac{\partial w_{i+1}}{\partial w_i}\cdot \frac{dw_i}{dx_1}\ \\
&&=\ w_i\cdot \frac{dx_{i+1}}{dx_1} + x_i \cdot \frac{dw_i}{dx_1}\ >\ 0,
\end{eqnarray*}

since, by the induction hypothesis, we have $w_i > d_i > 0$, $x_i \geq x_i' \geq 0$, $\frac{dw_i}{dx_1} > 0$, and we just proved that $\frac{dx_{i+1}}{dx_1} > 0$.

Now the other half of part 1, that $x_{i+1} \geq x_{i+1}'$ is immediate: this is because $x_{i+1}$ is a strictly increasing function of $x_1$. Similarly, part 2, that $w_{i+1} > d_{i+1}$ is also immediate: this is because $w_{i+1}$ is a strictly increasing function of $x_1$, and so $w_{i+1} \geq w_{i+1}' > d_{i+1}$ since $x_1'$ is a feasible point.
\end{proof}

We restate and prove Lemma~\ref{lem:curves}.
\begin{lemma}
In the feasible region for $x_1$, the upward curve for $x_n$ is strictly increasing, and the downward curve for $x_n$ is strictly decreasing.
\end{lemma}
\begin{proof}
Lemma~\ref{lem:increasing} establishes the fact that the upward curve for $x_n$ is strictly increasing.

Let $x_n' = \frac{w_{n-1} + d_n}{2w_{n-1}} = \frac{1}{2} + \frac{d_n}{2w_{n-1}}$ be the downward curve. Since $w_{n-1}$ is a strictly increasing function of $x_1$, we get that $x_n'$ is a strictly decreasing function of $x_1$.
\end{proof}

Finally, we restate and prove Theorem~\ref{lem:existence}.
\begin{theorem}[Existence]
A fixed point to the bargaining equations (\ref{eq:balance}) exists.
\end{theorem}
\begin{proof}
Consider the following function $f: [0, 1]^n \rightarrow [0, 1]^n$, which
represents a simultaneous update of the shares vector $\vec{x}$ on all edges
using the downward equations. Define $f$ as
$$f_i(x) = \min\left\{\frac{w_{i-1} + (1 - x_{i+1})d_i}{(2 - x_{i+1})w_{i-1}},\ 1\right\},$$
where $w_{i-1} = d_0\prod_{j=1}^{i-1} x_j$ as usual, and $x_{n+1} := 0$. We
make the convention that when $w_{i-1} = 0$, the first expression in the
minimum above is $+\infty$, so that $f_i(x) = 1$. This choice is consistent
with the limit as $w_{i-1} \rightarrow 0^+$, even in the case when $x_{i+1} =
1$. Thus, the first expression is a continuous function on the entire domain
$[0, 1]^n$. So $f$, being the minimum of two continuous functions, is also
continuous. Since $[0, 1]^n$ is a convex, compact set, by Brouwer's Fixed Point
Theorem, $f$ has a fixed point.

We show now that any fixed point of $f$ is a fixed point to the bargaining equations (\ref{eq:balance}). Suppose $x$ is a fixed point, i.e. $f(x) = x$. We
first show that for all $i$, we have $x_i < 1$. Suppose some $x_i = 1$. Then
$x_{i-1} = f_{i-1}(x) = \min\left\{\frac{w_{i-1}}{w_{i-1}}, 1\right\} = 1$.
Continuing inductively, we get that $x_j = 1$ for all $j \leq i$. This means
that $w_{i-1} = d_0\prod_{j=1}^{i-1} x_j = d_0$. Then
$$1\ =\ x_i\ =\ f_i(x)\ =\ \min\left\{\frac{d_0 + (1 - x_{i+1})d_i}{(2 - x_{i+1})d_0},\ 1\right\}\ =\ \frac{d_0 + (1 - x_{i+1})d_i}{(2 - x_{i+1})d_0}$$
since $d_i < d_0$. This implies that $x_{i+1} = 1$. Continuing inductively, we
get that $x_j = 1$ for all $j > i$, and hence, all $x_j = 1$, which implies
that all $w_j = 1$. But then, we get a contradiction for $x_n$, because using
the convention that $x_{n+1} = 0$ we get
$$1\ =\ x_n\ =\ f_n(x)\ =\ \min\left\{\frac{d_0 + d_n}{2d_0},\ 1\right\}\ \leq\ \frac{d_0 + d_n}{2d_0}\ <\ 1,$$
a contradiction. Hence no $x_i$ can equal $1$. So the fixed point satisfies
$x_i = \frac{w_{i-1} + (1 - x_{i+1})d_i}{(2 - x_{i+1})w_{i-1}}$ for all $i$, i.e. all the downward equations~(\ref{eq:downward}), and hence we
have a fixed point of the bargaining equations (\ref{eq:balance}).
\end{proof}

\section{Proof of Core Property}
\label{app:core}

The following easy result follows immediately for our game from the definition of the core:
\begin{proposition} \label{prop:offpath}
If $u$ is a payoff vector with $u_t > 0$ for some $t \notin p^*$, $u$ does not belong to the core.
\end{proposition}

\begin{proof}
The value generated by the coalition $p^*$ is $v^*$, so any payoff vector $u$ in the core must satisfy $\sum_{t \in p^*}u_t = v^*$. But $\sum_{s \in T} u_s = V(T) = v^*$. So if $u_t > 0$ for some $t \notin p^*$, we must have $\sum_{t \in p^*}u_t < v^*$ since all payoffs are nonnegative. So $u$ cannot belong to the core.
\end{proof}

This result immediately tells us that the {\em Shapley value} for this cooperative game does not belong to the core: the Shapley value can assign nonzero payoffs to nodes not on the winning path (Appendix A), and therefore need not belong to the core in our game.

Let $u^*$ denote the payoff vector resulting from the fixed point $\vec{x}$ to the system of bargaining equations. Recall that $u^*_t = 0$ for all nodes not on the path from $s_0$, the least common ancestor of the highest value leaves, to the root, and $u_t = (1-x_{ts})w_t$ where the $w_t$ are defined recursively as in Section \ref{s-model}.

Next we show the easy result that this payoff vector $u^*$, defined by the fixed point $\vec{x}$ of our bargaining problem, indeed belongs to the core. Note that here we need to consider the original tree-bargaining instance, and cannot simply argue about the reduced path bargaining instance, since the cooperative game is defined on the {\em tree} $T$ and the reduced path instance does not contain all the agents present in $T$.

We now restate and prove Theorem~\ref{thm:core}.
\begin{theorem}
The payoff vector $u^*$ belongs to the core of $(T,V)$.
\end{theorem}

\begin{proof}
By definition of the core, we only need to argue that $\sum_{i \in S} u_i \geq v(S)$ for all subsets of agents $S$. First, if $S$ does not contain a path $p$ from a leaf to the root, $V(S) = 0$. So trivially $\sum_{i \in S} u_i \geq v(S)$ for such $S$, since $u_i \geq 0$. Now if the set $S$ does contain such paths, it has value $V(S) = \max_{p\in S} V(p)$. Let $p$ be a path in $S$ with value $V(S)$.

First suppose that $V(S) < v^*$. Let $s$ be the least common ancestor of the leaves corresponding to $p^*$ and $p$. It is enough to show that the payoff to the nodes in the subpath of $p^*$ from $s$ to $r$ (\ie, the nodes in the intersection between $p$ and $p^*$) is at least $V(S)$.

Note that the sum of the payoffs to all nodes in the subpath from $s$ to the root is exactly the value $w_s$ that arrives at $s$, \ie, $\sum_{t\in p_{sr}} u_t = w_s$, since this subpath belongs to $p^*$.

Now, we must have $V(S) \leq d_s$, where the $d_i$'s are as defined in the reduction from the tree to the path. We have already shown that a fixed point satisfies $w_i \geq d_i$ for all nodes $i$ on the winning path $p^*$ in the proof of Lemma \ref{lem:maxwin}. Therefore,
\[
\sum_{i \in S} u_i \geq \sum_{i \in p_{sr}} u_i =  \sum_{i \in p^*_{sr}} u_i  = w_s \geq d_s \geq v(S),
\]
so that the core property is satisfied for the set $S$.

If $V(S) = v^*$, then $S$ must contain a path from a leaf $l$ with $v_l = v^*$. Let $s_0$ be the least common ancestor of all such leaves in $T$; from Lemma \ref{l-ties}, we know that $x_e = 1$ for all edges in the subtree rooted at $s_0$, and therefore $w_{s_0} = v^*$. As in the previous case, again, we have $\sum_{i \in S} u_i  \geq w_s = v^* = V(S)$, so that the core property is satisfied for such sets $S$ as well.
\end{proof}

\section{Proof of Strict Monotonicity}
\label{app:monotonicity}

We restate and prove Theorem~\ref{thm:monotonicity}:
\begin{theorem}
If any $d_i$ is increased (but is still kept less than $d_0$) while the remaining $d_{i'}$ are unchanged, the payoff of $i$
strictly increases.
\end{theorem}
\begin{proof}
We prove this theorem in two parts. In the first part (see Lemma~\ref{lem:below-increase}), we show that
increasing $d_i$ increases all $x_j$ for $j \leq i$. This implies that the
revenue reaching $i$ increases. In the second part (see Lemma~\ref{lem:above-decrease}), we show that increasing
$d_i$ decreases $x_{i+1}$, which implies that the fraction of revenue retained
at $i$, viz. $(1 - x_{i+1})$ goes up. We thus obtain our monotonicity
condition.
\end{proof}

We now state and prove Lemmas~\ref{lem:below-increase} and \ref{lem:above-decrease}. We fix some notation first. Let the new value of $d_i$ be $d_i' > d_i$ (but less than $1$). Let $\xstar_j$, ${\xstar_j}'$ be the fixed point shares on edge $e_j$ for the two sets of values respectively, and $\wstar_j$, ${\wstar_j}'$
the fixed point revenues reaching node $j$. Let $\xcirc_1$ and ${\xcirc_1}'$ be
the right end points for the feasible interval for $x_1$ for the two sets of
values. Let $x_n = U(x_1)$ and $x_n = D(x_1)$ denote the upward and downward
curves for $x_n$ computed as a function of $x_1$, and let $x_n' = U'(x_1')$ and
$x_n' = D'(x_1')$ be the corresponding curves for the second set of values.

\begin{lemma} \label{lem:below-increase}
${\xstar_j}' > \xstar_j$ for all $j \leq i$.
\end{lemma}
\begin{proof}
We now show that $x_1^\star < {x_1^\star}'$, which implies the lemma, since all
values $d_j$ for $j < i$ are unchanged, and by Lemma~\ref{lem:increasing},
increasing $x_1$ increases all shares above it up to $x_i$.

Let $x_1$ be a value that is feasible for both sets of values. Let $x_j$ and
$x_j'$ denote the shares on $e_j$ computed by the upward equations from $x_1$
with the value at node $i$ being $d_i$ and $d_i'$ respectively. Let $w_j$ and
$w_j'$ denote the corresponding revenues reaching node $j$. Since all values
$d_j$ for $j < i$ are unchanged, we have $x_j' = x_j$ (and consequently, $w_j'
= w_j$) for all $j \leq i$.

Now, since $d_i' > d_i$, the upward equation (\ref{eq:upward}) shows that $x_{i+1}' < x_{i+1}$. Since $w_i' = w_i$, Lemma~\ref{lem:increasing} implies that $x_j' < x_j$ for all $j > i$, and consequently, all $w_j' < w_j$ for $j > i$ as well. In particular, the upward curve $x_n = U(x_1)$ shifts down, and the downward curve, $x_n = D(x_1) = \frac{1}{2} + \frac{d_n}{2w_{n-1}}$, shifts up since $w_{n-1}$ increases.

Now, we have the following cases:
\begin{enumerate}
    \item $\xstar_1 < {\xcirc_1}'$. In this case, we are done since ${\xstar_1}' \geq {\xcirc_1}' > \xstar_1$.

    \item $\xstar_1 > {\xcirc_1}'$. Then $\xstar_1$ is feasible for both sets of values, and hence, by the above analysis, we have $U'(\xstar_1) < U(\xstar_1) = \xstar_n$ (because the upward curve shifts up), and $D'(\xstar_1) < D(\xstar_1) = \xstar_n$ (because the downward curve shifts down). Hence, the intersection point of the $U'$ and $D'$ curves must lie to the right of the intersection point of the $U$ and $D$ curves, which implies that ${\xstar_1}' > \xstar_1$.

    \item $\xstar_1 = {\xcirc_1}'$. We can fold this case into one of the above two cases, depending on whether ${\xcirc_1}'$ is feasible for the second set of values (then fold into case 2) or not (then fold into case 1).
\end{enumerate}
\end{proof}

\begin{lemma} \label{lem:above-decrease}
${\xstar_{i+1}}' < \xstar_{i+1}$.
\end{lemma}
\begin{proof}
We make use of the optimal sub-structure property of the fixed point solution:
i.e. suppose we fix the shares $x_j = \xstar_j$ for $j \leq i$, and recompute
the fixed point shares for the remaining tree. Then, the uniqueness property
implies that we recover the original fixed point shares. More precisely,
consider a bargaining problem on a smaller tree which is obtained by pruning
the original tree at node $i$, and setting the value at node $i$ to be
$\wstar_i$. The values $d_{i+1}, d_{i+2}, \ldots, d_n$ stay the same. Then, the
uniqueness of the fixed point implies that the fixed point shares on edges
$e_j$, for $j > i$, are precisely $\xstar_j$, because setting the shares to
$\xstar_j$ yields a fixed point solution to the new bargaining problem.

Note that by Lemma~\ref{lem:below-increase}, we have ${\wstar_i}' > \wstar_i$.
We exploit the optimal sub-structure property as follows. If we prune the
subtree rooted at $i$, then a fixed point to the bargaining problem where the
maximum bid, at node $i$, is $\wstar_i$ has the same shares on the edges above
$i$ as in the original instance. If we increase this maximum bid to
${\wstar_i}'$, then the new fixed point has the same shares as the fixed point
on the edges above $i$ as in instance where $d_i$ is increased to $d_i'$. We
now want to show that $x_{i+1}' < x_{i+1}$.

Mapping this reduction back to our standard notation for the path bargaining
problem, we now need to prove the following: in a path bargaining problem,
suppose the maximum bid value $d_0$ is increased to some value greater than its
original value. Then, in the new fixed point solution, $x_1$ decreases. We now
prove such a claim.

First, we prove that if $x_1 \in [0,1]$ is feasible for the original set of values, then it is feasible when the value $d_0$ is increased to $d_0'$, for any $d_0' > d_0$.
Let the new share values computed from $x_1$ be $x_i'$, and the corresponding
revenues $w_i'$. We prove by induction on $i$ that these values are feasible,
and furthermore, for all $i$, we have $x_i' \geq x_i$ and $w_i' > w_i$.

For $i = 1$, we have $x_1' = x_1 \in [0, 1]$, and $w_1' = x_1'w_0' = x_1d_0' > x_1d_0 = w_1$. Finally, $w_1' \geq w_1 > d_1$, since $x_1$ is feasible for the original set of values. Thus, the base case is established. Assume that the inductive hypothesis is true for some $i \geq 1$. Now we prove it for $i+1$.

We imagine the change from $x_{i+1}$ to $x_{i+1}'$ as happening in two steps.
First, we increase $w_{i-1}$ to $w_{i-1}'$, keeping the share on $e_i$ fixed to
$x_i$. This gives an intermediate share on $e_{i+1}$ of value $x_{i+1}'' = 1 -
\frac{w_{i-1}'(1 - x_i)}{w_{i-1}'x_i - d_i}$. Then, we increase the share on
$e_i$ from $x_i$ to $x_i'$, keeping the revenue at node $i$ fixed to
$w_{i-1}'$. We show that in each step, the share on $e_{i+1}$ increases and as
is bounded in $[0, 1]$, i.e. $x_{i+1} \leq x_{i+1}'' \leq x_{i+1}'$, and
$x_{i+1}'' \in [0,1]$ and $x_{i+1}' \in [0,1]$.

In the first step, we have
\begin{eqnarray*}
&&x_{i+1}'' - x_{i+1}'\ =\ \left[1 - \frac{w_{i-1}'(1 - x_i)}{w_{i-1}'x_i - d_i}\right] - \left[1 - \frac{w_{i-1}(1 - x_i)}{w_{i-1}x_i - d_i}\right]\ \\
&&=\ \frac{(1-x_i)d_i(w_{i-1}' - w_{i-1})}{(w_{i-1}'x_i - d_i)(w_{i-1}x_i - d_i)}\ \geq\ 0.
\end{eqnarray*}

Furthermore, since $\frac{w_{i-1}'(1 - x_i)}{w_{i-1}'x_i - d_i} \geq 0$ (because $w_{i-1}'x_i \geq w_{i-1}x_i = w_i > d_i$ by feasibility of $x_1$), we have $x_{i+1}'' = 1- \frac{w_{i-1}'(1 - x_i)}{w_{i-1}'x_i - d_i} \leq 1$. Thus, $x_{i+1}'' \in [0, 1]$.

Now, we prove the second step. We have

\begin{eqnarray*}
&&x_{i+1}' - x_{i+1}'' = \left[1 - \frac{w_{i-1}'(1 - x_i')}{w_{i-1}'x_i' - d_i}\right] - \left[1 - \frac{w_{i-1}'(1 - x_i)}{w_{i-1}'x_i - d_i}\right]\ \\
&&=\ \frac{w_{i-1}'(w_{i-1}' - d_i)(x_i' - x_i)}{(w_{i-1}'x_i' - d_i)(w_{i-1}'x_i - d_i)}\ \geq\ 0,
\end{eqnarray*}

since $w_{i-1}' \geq w_{i-1} > d_i$ by feasibility of $x_1$. Furthermore, since $\frac{w_{i-1}'(1 - x_i')}{w_{i-1}'x_i' - d_i} \geq 0$, we have $x_{i+1}' = 1- \frac{w_{i-1}'(1 - x_i')}{w_{i-1}'x_i' - d_i} \leq 1$. Thus, $x_{i+1}' \in [0, 1]$, as required.

As for $w_{i+1}'$: we have $w_{i+1}' = x_{i+1}'w_i' > x_{i+1}w_i = w_{i+1}$
where the strict inequality follows by the inductive hypothesis. Furthermore,
we have $w_{i+1}' > w_{i+1} > d_{i+1}$ by feasibility of $x_1$. Thus,
$x_{i+1}'$ is a feasible share value, and the induction is complete.

Now, we argue as in Lemma~\ref{lem:below-increase}. When the max value $d_0$
increases to $d_0' > d_0$, we conclude from the analysis above that the upward
curve for the share on $e_n$, viz. $x_n' = U'(x_1)$ shifts upwards from the
original upward curve $x_n = U(x_1)$. Also, the downward curve $x_n' = D'(x_1)
= \frac{1}{2} + \frac{d_n}{2w_{n-1}'}$ shifts {\em strictly} downwards from the
original downward curve $x_n = D(x_1)$ since $w_n' > w_n$. Now consider
$\xstar_1$. We have $U'(\xstar_1) \geq U(\xstar_1) = \xstar_n$, and
$D'(\xstar_1) < D(\xstar_1) = \xstar_n$. So $U'(\xstar_1) > D'(\xstar_1)$, and
since $U'$ and $D'$ are increasing and decreasing respectively, their
intersection point must satisfy ${\xstar_1}' < \xstar_1$.
\end{proof}

\section{Computability of Fixed Point}
\label{app:computability}

Before going into algorithms, we need to first get a quantitative version of
the arguments of Section \ref{sec:fixedpoint} in order to understand what
tolerance to use in our algorithms.

\begin{lemma} \label{lem:lowerbounds}
For any $i$, we have $x_i^\star \geq \frac{1}{2 - x_{i+1}}$. In particular, $x_i^\star \geq \frac{n-i+1}{n-i+2}$, and $w_i^\star \geq \frac{n-i+1}{n+1}$.
\end{lemma}
\begin{proof}
This follows immediately from the downward equations:
$$x_i\ =\ \frac{w_{i-1} + (1 - x_{i+1})d_i}{(2 - x_{i+1})w_{i-1}}\ \geq\ \frac{1}{2 - x_{i+1}}.$$ The other parts follow by induction, starting with $x_{n+1}^\star = 0$.
\end{proof}

Recall that $\gamma = \min\{1 - \max_{i > 0} \{d_i\}, \frac{1}{n+2}\}$.

\begin{lemma}\label{lem:comp1}
For all $i$ we have $x_i^\star \leq 1 - \gamma^{4n}$ and $w_i^\star \geq d_i + \gamma^{4n+1}$.
\end{lemma}
\begin{proof}
For notational convenience, we drop the $\star$ superscript. Suppose $x_i > 1 - \gamma^{4n}$ for some $i$. Then we get a contradiction as follows. First, we have
$$x_{i-1} = \frac{w_{i-1} + (1 - x_{i})d_{i-1}}{(2 - x_{i})w_{i-1}} > \frac{1}{2 - x_i} > \frac{1}{1 + \gamma^{4n}} > 1 - \gamma^{4n}.$$
Inductively, we get that for all $j \leq i$, $x_j > 1 - \gamma^{4n}$.

Hence, $w_{i-1} = \prod_{j=1}^{i-1} x_j > 1 - n\gamma^{4n}$. Then,
$$1 - x_{i+1}\ =\ \frac{(1 - x_i)w_{i-1}}{x_iw_{i-1} - d_i}\ <\ \frac{\gamma^{4n}}{(1 - n\gamma^{4n}) - (1 - \gamma)}\ \leq\ \gamma^{4n - 2}$$
since $\gamma \leq \frac{1}{n+2}$. Hence, we get that $x_{i+1} > 1 - \gamma^{4n-2}$. Continuing inductively, we get that
$x_{n-1} > 1 - \gamma^{(4n - 2(n-1-i))} > 1 - \gamma^{2n}$. Then, we have $w_{n-1} > 1 - n\gamma^{2n}$. But then using the downward equation for $x_n$, we get
\begin{eqnarray*}
&&1 - \gamma^{2n} < x_n = \frac{1}{2} + \frac{d_n}{2w_{n-1}} < \frac{1}{2} + \frac{1 - \gamma}{2(1 - n\gamma^{2n})} \\
&& < 1 - \frac{1}{2}\gamma + \frac{1}{2}n\gamma^{2n}
\quad \Longrightarrow \quad \frac{1}{n+2} < \gamma^{2n-1} \leq \gamma,
\end{eqnarray*}
which is a contradiction.

So we have established that all $x_i \leq 1 - \gamma^{4n}$. Now suppose for some $i$, we have $w_i < d_i + \gamma^{4n+1}$. From Lemma~\ref{lem:lowerbounds}, we have $w_{i-1} \geq \frac{1}{n+1} > \gamma$. So, we get from the upward equation for $x_{i+1}$:
$$1\ \geq\ 1 - x_{i+1}\ =\ \frac{(1 - x_i)w_{i-1}}{x_iw_{i-1} - d_i}\ > \frac{\gamma^{4n}\cdot \gamma}{\gamma^{4n+1}}\ =\ 1,$$
a contradiction.
\end{proof}

Now we can give an algorithm to approximate $x_1^\star$ very accurately:
\begin{lemma} \label{lem:computation}
Given any $\epsilon > 0$, there is an algorithm that runs in $\text{poly}(n, \log(1/\epsilon))$ time which computes a value $x_1$ such that $|x_1 - x_1^\star| < \epsilon$.
\end{lemma}
\begin{proof}
The algorithm basically runs binary search using the upward and downward curves for $x_n$ to find an approximation to $x_1^\star$. Care needs to be taken to make sure we are in the feasible region $[x_1^\circ, 1]$ for $x_1$. This can be incorporated in the binary search for $x_1^\star$ by moving to the right half of the current interval whenever the current value of $x_1$ is infeasible (detected by a violation of the conditions $x_i \in [0, 1]$ and $w_i > d_i$.

Once we are in the feasible region, we know that the upward and downward curves for $x_n$ intersect, and hence binary search can proceed by comparing the values of the two curves; moving left if the upward curve is higher than the downward curve, and right otherwise. The pseudocode is given in Algorithm~\ref{alg:fixedpoint}. The running time follows from the properties of binary search.
\end{proof}

\begin{algorithm}[!ht] 
\caption{FixedPoint$(d, \epsilon)$} 
	\begin{algorithmic}[1]
        \STATE Set $\ell = 0$, $h = 1$.

        \WHILE{$|h - \ell| > \epsilon$}

        \STATE $x_1 = \frac{\ell + h}{2}.$

        \STATE Compute $x_2, \ldots, x_n$ using the upward equations (\ref{eq:upward}). Also compute $w_1, w_2, \ldots, w_n$.

        \IF{any $x_i \notin [0, 1]$ or any $w_i \leq d_i$}
            \STATE $\ell = x_1$
        \ELSE
            \STATE Compute $x_n' = \frac{1}{2} + \frac{d_n}{2w_n}$.
            \IF{$x_n > x_n'$}
                \STATE Set $h = x_1$
            \ELSE
                \STATE Set $\ell = x_1$
            \ENDIF
        \ENDIF
        \ENDWHILE
        \STATE Set $x_1 = h$, and compute other $x_i$ using this value of $x_1$. Return this solution.
	\end{algorithmic}
	\label{alg:fixedpoint} 
\end{algorithm}

We show now that an $\epsilon$-fixed point can be computed efficiently as well. This proves Theorem~\ref{thm:eps-fixed point}, which we restate here for convenience:
\begin{theorem}
Given any $\epsilon > 0$, there is an algorithm that runs in $\text{poly}(n, \log(1/\gamma\epsilon))$ time which computes an $\epsilon$-fixed point to the bargaining equations (\ref{eq:balance}).
\end{theorem}
\begin{proof}
The idea is to compute $x_1^\star$ accurately enough so that for all $i$, we have $|x_i - x_i^\star| < \epsilon$, and $|w_i - w_i^\star| < \epsilon$. Since all these quantities are in $[0, 1]$, such an approximation implies that the balance condition is satisfied up to an additive error of $10\epsilon$ (the constant $10$ here is a crude estimate).

Suppose we estimate $x_1^\star$ to accuracy $\epsilon_1$. Then we compute bounds on the accuracy $\epsilon_i$ to which $x_i^\star$, is computed using the upward equations. To do this, we also need bounds $\delta_i$ on the accuracy to which $w_{i-1}^\star$ is computed. From the upward equations, using the fact that the denominator $w_i^\star - d_i > \gamma^{4n+1}$, we get the following recurrence relation:
$$\epsilon_{i+1}\ \leq\ \frac{c(\epsilon_i + \delta_i)}{\gamma^{8n+2}}\ \leq\ \frac{(\epsilon_i + \delta_i)}{\gamma^{10n}} $$
for some constant $c$ obtained by looking at the number of arithmetic operations performed. The last inequality follows by choosing $\gamma < 1/c$.

As for $\delta_{i+1}$, we get the following recurrence relation:
$$\delta_{i+1}\ \leq\ c(\epsilon_i + \delta_i)\ \leq\ \frac{(\epsilon_i + \delta_i)}{\gamma^{10n}},$$
where for convenience of notation we use the same constant $c$. Putting these together, and setting $\eta_i = \epsilon_i + \delta_i$, we get
$$\eta_{i+1}\ \leq\ \frac{1}{\gamma^{10n}}\eta_i,$$
and so for all $i$, we get
$$\eta_i\ \leq\ \frac{1}{\gamma^{10n^2}}\eta_1.$$ Note that $\eta_1 = \epsilon_1 + \delta_1 = \epsilon_1$. So by choosing $\epsilon_1 = \gamma^{10n^2}\epsilon$ we get that $\eta_i \leq \epsilon$, and then the solution computed is a $10\epsilon$-fixed point. The running time for the binary search algorithm
is  $\text{poly}(n, \log(1/\epsilon_1)) = \text{poly}(n,
\log(1/\gamma\epsilon))$.
\end{proof}

\end{document}